\documentclass[conference]{IEEEtran}

\IEEEoverridecommandlockouts
\usepackage{cite}
\usepackage{amsmath,amssymb,amsfonts,amsthm}
\usepackage{algorithm,algorithmic}
\usepackage{graphicx}
\usepackage{textcomp}
\usepackage{xcolor}
\usepackage[acronym]{glossaries}
\usepackage{blindtext}
\def\BibTeX{{\rm B\kern-.05em{\sc i\kern-.025em b}\kern-.08em
    T\kern-.1667em\lower.7ex\hbox{E}\kern-.125emX}}
\usepackage{booktabs}
\usepackage{bm}
\usepackage{import}
\usepackage{physics}
\usepackage[hidelinks]{hyperref}
\usepackage{url}
\usepackage{multirow}
\usepackage{nicefrac}
\usepackage{caption}

\newtheorem{definition}{Definition}
\newtheorem{proposition}{Proposition}
\newtheorem{theorem}{Theorem}

\definecolor{col_g080}{rgb}{1.0, 0.5, 0.0}
\definecolor{col_g085}{rgb}{0.0, 0.5, 0.0}
\definecolor{col_g090}{rgb}{0.0, 0.0, 1.0}
\definecolor{col_g095}{rgb}{1.0, 0.0, 0.0}

\definecolor{col_warm}{rgb}{0.0, 0.5, 0.0}
\definecolor{col_cold}{rgb}{1.0, 0.0, 0.0}

\definecolor{col_per00}{rgb}{1.0, 0.5, 0.0}
\definecolor{col_per01}{rgb}{0.0, 0.5, 0.0}
\definecolor{col_per03}{rgb}{0.0, 0.0, 1.0}

\usepackage{pgfplots}
\pgfplotsset{compat=1.14}
\usepgfplotslibrary{fillbetween}

\usepackage[raggedright]{subfigure}

\usepackage[capitalize,noabbrev]{cleveref}
\Crefname{equation}{Eq.}{Eqs.}
\Crefname{figure}{Fig.}{Figs.}
\Crefname{section}{Sec.}{Secs.}
\Crefname{table}{Tab.}{Tabs.}
\Crefname{appendix}{App.}{Appx.}
\Crefname{definition}{Def.}{Defs.}
\Crefname{proposition}{Prop.}{Props.}
\Crefname{theorem}{Thm.}{Thms.}

\usepackage{tikz}
\usepackage{textcomp}
\usepackage{hyperref}
\usepackage{lipsum}
\newcommand\copyrighttext{%
  \footnotesize \textcopyright 2024 IEEE. Personal use of this material is permitted.
  Permission from IEEE must be obtained for all other uses, in any current or future
  media, including reprinting/republishing this material for advertising or promotional
  purposes, creating new collective works, for resale or redistribution to servers or
  lists, or reuse of any copyrighted component of this work in other works.}
\newcommand\copyrightnotice{%
\begin{tikzpicture}[remember picture,overlay]
\node[anchor=south,yshift=10pt] at (current page.south) {\fbox{\parbox{\dimexpr\textwidth-\fboxsep-\fboxrule\relax}{\copyrighttext}}};
\end{tikzpicture}%
}

\begin{document}

\usetikzlibrary{calc}

\newacronym{vqc}{VQC}{variational quantum circuit}
\newacronym{vqa}{VQA}{variational quantum algorithm}
\newacronym{rl}{RL}{reinforcement learning}
\newacronym{qrl}{QRL}{quantum reinforcement learning}
\newacronym{ml}{ML}{machine learning}
\newacronym{qml}{QML}{quantum machine learning}
\newacronym{dnn}{DNN}{deep neural network}
\newacronym{qc}{QC}{quantum computing}
\newacronym{nisq}{NISQ}{noisy intermediate-scale quantum}
\newacronym{qpg}{QPG}{quantum policy gradient}
\newacronym{mdp}{MDP}{Markov Decision Process}
\newacronym{qpi}{QPI}{quantum policy iteration}
\newacronym{lse}{LSE}{linear system of equations}
\newacronym{varqpi}{VarQPI}{variational quantum policy iteration}
\newacronym{wsvarqpi}{WS-VarQPI}{warm-start VarQPI}
\newacronym{svd}{SVD}{singular value decomposition}

\title{Warm-Start Variational Quantum Policy Iteration
    \thanks{The research was supported by the Bavarian Ministry of Economic Affairs, Regional Development and Energy with funds from the Hightech Agenda Bayern via the project BayQS.\\
    $^{\ast}$These authors contributed equally.\\
    Correspondence to: nico.meyer@iis.fraunhofer.de}
}

\author{
    \IEEEauthorblockN{
        Nico Meyer$^{\ast,1,2}$, Jakob Murauer$^{\ast,1}$, Alexander Popov$^{1}$, Christian Ufrecht$^{1}$\\Axel Plinge$^{1}$, Christopher Mutschler$^{1}$, and Daniel D. Scherer$^{1}$
    }
    \IEEEauthorblockA{
        $^{1}$Fraunhofer IIS, Fraunhofer Institute for Integrated Circuits IIS, Nürnberg, Germany\\
        $^{2}$Pattern Recognition Lab, Friedrich-Alexander-Universität Erlangen-Nürnberg, Erlangen, Germany
    }
}

\maketitle

\begin{abstract}
    Reinforcement learning is a powerful framework aiming to determine optimal behavior in highly complex decision-making scenarios. This objective can be achieved using policy iteration, which requires to solve a typically large linear system of equations. We propose the \gls{varqpi} algorithm, realizing this step with a NISQ-compatible quantum-enhanced subroutine. Its scalability is supported by an analysis of the structure of generic reinforcement learning environments, laying the foundation for potential quantum advantage with utility-scale quantum computers. Furthermore, we introduce the warm-start initialization variant (WS-VarQPI) that significantly reduces resource overhead. The algorithm solves a large \texttt{FrozenLake} environment with an underlying $256 \times 256$-dimensional linear system, indicating its practical robustness.
\end{abstract}

\copyrightnotice  
\glsresetall
\vspace{-0.4cm}  
\section{Introduction.}

\Gls{qrl}~\cite{Meyer_2022a} is an active field of research at the intersection of \gls{qc} and machine learning. The success of classical \gls{rl}~\cite{Jumper_2021,Ruiz_2024} can mostly be attributed to the use of large neural networks as function approximators for the \gls{rl} setup. The advent of \gls{nisq}~\cite{Preskill_2018} devices in the last decade has spurred the development of so-called variational approaches~\cite{Bharti_2022}. The typical idea in the realm of \gls{qrl} is to use \glspl{vqc} for expressing either the intended behavior directly~\cite{Jerbi_2021,Meyer_2023a}, or approximate an intermediate value function~\cite{Chen_2020,Skolik_2022}. Extensions include combining multiple function approximators~\cite{Wu_2023}, and advanced training routines~\cite{Meyer_2023b} for the underlying \glspl{vqc}. While these approaches are quite flexible, they make it difficult to exploit structure of the considered problem, blurring the path to potential quantum advantage.

A more structured approach to \gls{rl} with dynamic programming is based on so-called policy iteration~\cite{Howard_1960,Lagoudakis_2003}. However, directly solving the associated large \gls{lse} is typically intractable with classical methods -- requiring to resort to iterative solution methods~\cite{Sutton_2019}. This suggests the application of quantum-enhanced linear algebra routines, which provide exponential speed-ups for solving a \gls{lse} under certain conditions~\cite{Harrow_2009,Low_2019}. Unfortunately, executing those routines requires large-scale fault-tolerant quantum devices, which are currently not available. \gls{nisq}-compatible variational algorithms for linear algebra have been proposed~\cite{Bravo_2023,Xu_2021} -- with empirical results suggesting scaling behavior~\cite{Bravo_2023} competitive with HHL-based approaches. 

\begin{figure}
    \centering
    \def\svgwidth{0.98\linewidth}
    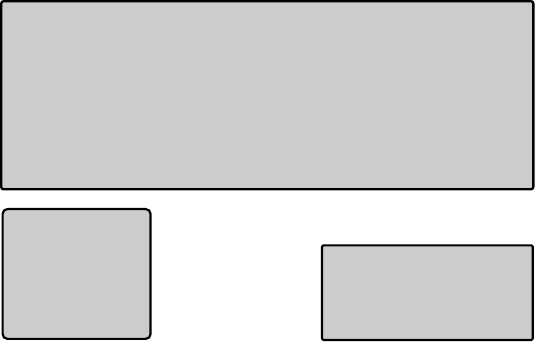
    \captionsetup{font=footnotesize}
    \caption{Proposed routine of variational quantum policy iteration. The evaluation of the state-action value function $Q_{\pi}$ is formulated as a \gls{lse} $A \bm{x} = \bm{b}$. The solution is computed variationally~\cite{Bravo_2023}, where the blue sub-circuit prepares a state proportional to $Q_{\pi}$. Classically efficient $\ell_{\infty}$-tomography is used for policy improvement. This procedure is iterated until an optimal policy $\pi_*$ is found. After random initialization, warm-start variational parameters $\bm{\alpha}$ are carried over from the previous iteration, leading to faster policy evaluation.}
    \label{fig:vqpi}
\end{figure}

\textbf{Contribution.}
We introduce the \gls{varqpi} algorithm, using a variational linear system solver to perform the quantum-enhanced policy evaluation step within a policy iteration framework. The system matrices associated with generic \gls{rl} environments are shown to be well-behaved with respect to quantum-enhanced linear algebra routines. We extend the algorithm with a warm-start parameter initialization approach to \gls{wsvarqpi}, which significantly speeds up convergence. The efficiency of the algorithm is demonstrated on different instances of \texttt{FrozenLake} environments, supported with an extensive analysis of the resource requirements.

\textbf{Related Work.} 
The application of fault-tolerant linear system solvers to \gls{qpi} has been studied by Cherrat et al.~\cite{Cherrat_2023}. We focus our analysis on the \gls{nisq}-compatible regime, employing a variational \gls{lse} solver proposed in~\cite{Bravo_2023} and implemented in the \texttt{variational-lse-solver} library~\cite{Meyer_2024}. The proposed algorithm uses direct policy evaluation (instead of the typical iterative approach)~\cite{Sutton_2019,Lagoudakis_2003}, combined with classically efficient $\ell_{\infty}$ tomography~\cite{Kerenidis_2019,Cherrat_2023} for policy improvement.

This work is structured as follows: In~\cref{sec:method} we introduce preliminaries and describe (warm-start) \gls{varqpi} in detail. An experimental demonstration of \gls{varqpi} and \gls{wsvarqpi} is conducted in~\cref{sec:experiments}, completed with experiments on a large \texttt{FrozenLake} environment. An analysis of resource requirements for the quantum-enhanced routine is conducted in~\cref{sec:resource_analysis}. We summarize the results, highlight bottlenecks, and pinpoint potential for future work in~\cref{sec:conclusion}. We study the \gls{lse} induced by typical \gls{rl} environments regarding requirements on sparseness and condition number in the appendix.

\section{\label{sec:method}Preliminaries and Method Description}

\gls{rl} is an algorithmic framework to solve complex time-dependent decision-making tasks. The underlying \gls{mdp} can be expressed as a $5$-tuple ($\mathcal{S}, \mathcal{A}, r, p, \gamma)$, with state set $\mathcal{S}$, and action set $\mathcal{A}$. The reward function $r(s,a,s') \mapsto \mathbb{R}$ assigns a scalar value to executing action $a$ in state $s$ and transitioning to state $s'$. The underlying dynamics are described by $p(s'|s,a) \mapsto \left[ 0, 1 \right]$, with $\sum_{s'} p(s'|s,a) = 1$ for all $s \in \mathcal{S}, ~a \in \mathcal{A}$. The discount factor $0 \leq \gamma \leq 1$ controls the weighting of immediate and long-term rewards. The goal is to derive a policy $\pi(a|s) \mapsto \left[ 0, 1 \right]$, that maximizes the discounted long-term reward $G_t \gets \sum_{t'=t}^{\infty} \gamma^{t'-t} r_{t'}$ over discrete timesteps $t$, with $r_{t'}$ following the reward function.

\subsection{Direct Policy Iteration}

To maximize the long-term reward, one can start off with the Bellman equation for the state-action value function~\cite{Sutton_2019}:
{\small\begin{align}
    \label{eq:bellman_equation}
    q_{\pi}(s,a) &= \sum_{s'} p(s'|a,s) \left[ r(s,a,s') + \gamma \sum_{a'} \pi(a'|s') q_{\pi}(s',a') \right]
\end{align}}
Solving this equation for a given $\pi$ is referred to as \emph{policy evaluation}, with existence of a solution being guaranteed for $\gamma < 1$~\cite{Sutton_2019}. Assuming a solution to~\cref{eq:bellman_equation} for now, one can perform \emph{greedy policy improvement} as
\begin{align}
    \label{eq:greedy_policy}
    \pi(s) \gets \mathrm{argmax}_a ~q(s,a),
\end{align}
with $\pi(a|s) = 1$, if $a = \pi(s)$, and $\pi(a|s) = 0$ otherwise. The repeated execution of policy evaluation and improvement is referred to as \emph{policy iteration}. It is guaranteed to converge to an optimal $\pi_*$ and typically does so very fast~\cite{Sutton_2019}.

Evaluating~\cref{eq:bellman_equation} can be expressed as solving a \gls{lse} with $\abs{\mathcal{S}} \cdot \abs{\mathcal{A}}$ equations and unknowns. In practice -- due to the large system size -- one often resorts to iterative solution methods~\cite{Sutton_2019}. Using matrix notation for~\cref{eq:bellman_equation}, we get 
\begin{align}
    \label{eq:bellman_equation_matrix}
    Q_{\pi} = R + \gamma P \Pi Q_{\pi},
\end{align}
with state-action vector $Q_{\pi} = q(s_i,a_k)_{(ik)} \in \mathbb{R}^{\abs{\mathcal{S}} \times \abs{\mathcal{A}}}$, reward vector $R =  \sum_{s'} p(s'|s_i,a_k)r(s_i,a_k,s')_{(ik)} \in \mathbb{R}^{\abs{\mathcal{S}} \times \abs{\mathcal{A}}}$,  transition probability matrix $P = p(s'_j|s_i,a_k)_{(ik),j} \in \mathbb{R}^{\abs{\mathcal{S}} \abs{\mathcal{A}} \times \abs{\mathcal{S}}}$, and policy matrix $\Pi = \pi(a'_l|s'_j)_{j,(jl)} \in \mathbb{R}^{\abs{\mathcal{S}} \times \abs{\mathcal{S}} \abs{\mathcal{A}}}$. Furthermore, \cref{eq:bellman_equation_matrix} can be re-formulated to
\begin{align}
    \label{eq:bellman_equation_lse}
    (\mathbb{I} - \gamma P \Pi) Q_{\pi} &= R,
\end{align}
which can be interpreted as a \gls{lse} with system matrix $A_{\pi} := \mathbb{I} - \gamma P \Pi$, right-hand side $R$, and sought-after solution $Q_{\pi}$. 
The generalization of our results to model-free scenarios with function approximation -- also referred to as least-squares policy iteration~\cite{Lagoudakis_2003} -- is left for future work.

\subsection{\label{subsec:lse}Variational LSE Solvers}

Let the system matrix from~\cref{eq:bellman_equation_lse} be $A \in \mathbb{R}^{N \times N}$, where $N := \abs{\mathcal{S}} \cdot \abs{\mathcal{A}}$. Classical solution methods scale polynomially in the system size and might become infeasible for large state and action spaces. Quantum-enhanced methods like HHL~\cite{Harrow_2009} or modern variants~\cite{Gilyen_2019,Low_2019} promise poly-logarithmic scaling in the system size $N$. The use of these techniques in the context of quantum policy iteration has been researched by Cherrat et al.~\cite{Cherrat_2023}. Unfortunately, executing these algorithms requires large-scale and fault-tolerant quantum devices. An alternative for the \gls{nisq}-era is the use of corresponding variational algorithms~\cite{Cerezo_2021}. While these typically do not provide explicit guarantees of quantum advantage, empirical speed-ups in certain scenarios have been observed. We focus on a variational \gls{lse} solver proposed by Bravo-Prieto et al.~\cite{Bravo_2023}.

In its most abstract form, the algorithm is designed for preparing a state $\ket{x}$, that is proportional to the solution of $A\bm{x}=\bm{b}$. To enable training, one can use a parameterized ansatz $\ket{x(\bm{\alpha})}$. Furthermore, it must be possible to prepare a normalized version of $\bm{b}$, i.e. a state $\ket{b}$. While $A$ is not unitary in general, it is always possible to find a unitary decomposition
\begin{align}
    \label{eq:unitary_decomposition}
    A = \sum_{k=0}^{L-1} c_k A_k,
\end{align}
where $c_k$ are complex coefficients~\cite{Zhan_2001} -- see also~\cref{app:decomposition}. With the in general not normalized state $\ket{\psi} = A \ket{x(\bm{\alpha})}$ and circuit realizations of the unitaries $A_k$, it is possible to evaluate a loss function
\begin{align}
    \label{eq:loss}
    C_G = 1 - \abs{\expval{b | \Psi}}^2,
\end{align}
with $\ket{\Psi} = \ket{\psi}/\expval{\psi|\psi}$. The loss vanishes for $\ket{\psi}$ proportional to $\ket{b}$, and therefore can be used to drive the parameters $\bm{\alpha}$ towards this objective. The use of Hadamard and Hadamard-overlap tests allow the evaluation of the loss function on quantum hardware (see e.g.~\cite{Bravo_2023}). Furthermore,~\cref{eq:loss} is differentiable, which enables the use of gradient-based optimization techniques. Apart from that, also a local version of the cost function exists, which alleviates training~\cite{Bravo_2023}. The routine for \gls{varqpi} implemented in this paper is based on the library \texttt{variational-lse-solver}, which contains multiple quantum-enhanced \gls{lse} solvers~\cite{Meyer_2024}.

An bound on the required accuracy is given by the operational meaning of the loss function $C_G$, reading as
\begin{equation}
    \label{eq:threshold}
    C_G \geq \frac{\epsilon^2}{\kappa^2}
\end{equation}
for a desired trace distance $\epsilon$ between optimal and approximate solution\cite{Bravo_2023}. Here, $\kappa = \frac{\sigma_{\mathrm{max}}}{\sigma_{\mathrm{min}}}$ denotes the condition number of the system matrix $A$, where $\sigma_{\mathrm{max}}(A) := \max_{\| x \| = 1} \| A x \|$ and $\sigma_{\mathrm{min}}(A) := \min_{\| x \| = 1} \| A x \|$ denote the largest and smallest singular values, respectively. We demonstrate in~\cref{app:condition_number} that the systems associated with typical \gls{rl} environments are well-conditioned. The results in~\cref{subsec:threshold} suggest that~\cref{eq:threshold} is not tight, allowing for much faster convergence in practice. Furthermore, matrix sparsity -- implicitly assumed throughout the work of Bravo-Prieto et al.~\cite{Bravo_2023} -- is investigated in~\cref{app:sparsity}.




\subsection{Variational Quantum Policy Iteration with Warm Start}

The routine for variational \gls{varqpi} combines quantum-enhanced policy evaluation with $\ell_{\infty}$-tomography and classical policy improvement~\cite{Cherrat_2023,Kerenidis_2019} to the following algorithm: For a random initial (deterministic) policy $\pi$, one variationally prepares a state $\ket{Q_{\pi}} \equiv \ket{x(\bm{\alpha}_*)}$, such that
\begin{align}
    \label{eq:var_policy_evalauation}
    \ket{Q_{\pi}} &\approx \frac{1}{\| Q_{\pi} \|} \begin{bmatrix} Q_{\pi}(s=0,a=0) \\ Q_{\pi}(s=0,a=1) \\ \vdots \\ Q_{\pi}(s=1,a=0) \\ \vdots \end{bmatrix} \\
    &\sim (\mathbb{I} - \gamma P \Pi)^{-1} R 
\end{align}
where $\bm{\alpha}_*$ are the parameters optimized using the algorithm described in~\cref{subsec:lse} for the system in~\cref{eq:bellman_equation_lse}. Subsequently, multiple preparations of $\ket{Q_{\pi}}$ are used to acquire counts $M(s,a)$. The raw bitstring measurements are associated with $(s,a)$ following~\cref{eq:var_policy_evalauation}. Policy improvement is furthermore conducted as
\begin{align}
    \label{eq:greedy_action_selection}
    \pi(s) \gets \mathrm{argmax}_a M(s,a).
\end{align}
Since policy iteration converges with respect to the $\ell_{\infty}$ norm, it has been shown~\cite{Cherrat_2023} that $M(s,a)$ can be estimated using $\ell_{\infty}$-tomography (see e.g.~\cite{Kerenidis_2019}). This requires only a polynomial number of samples, even though $M(s,a)$ is of exponential size in the number of qubits.

This application of quantum policy evaluation and classical policy improvement is repeated iteratively. The initial parameters for the first policy evaluation step $\bm{\alpha}_{\mathrm{init}}^{0}$ are drawn uniformly at random from $\left[ 0, 2 \pi \right]$, and trained to $\bm{\alpha}_*^{0}$ (up to an arbitrary small error threshold). Due to policy iteration updating the value function and associated greedy policy step by step, the greedy actions between successive cycles only change for a limited number of states. Therefore, we assume that $| Q_{\pi^{0}} \rangle \equiv |x(\bm{\alpha}_*^{0})\rangle$ is reasonably close to $\ket{Q_{\pi^{1}}}$. Similarly, $\ket{Q_{\pi^{i}}}$ should be close to $\ket{Q_{\pi^{i+1}}}$ for each two consecutive iterations. Exploiting this, we formulate \gls{wsvarqpi} by initializing the variational parameters
\begin{align}
    \bm{\alpha}_{\mathrm{init}}^{i+1} \gets \bm{\alpha}_*^{i},
\end{align}
with $\bm{\alpha}_{\mathrm{init}}^{0}$ selected randomly as described above. This procedure is also sketched in~\cref{fig:vqpi}. Policy evaluation and improvement is repeated to convergence, i.e.\ until $\pi^{i+1}(s) = \pi^{i}(s)$ for all non-terminal states $s$. We demonstrate in~\cref{subsec:warm_start}, that the proposed \gls{wsvarqpi} strategy in fact enhances the resource efficiency compared to standard \gls{varqpi}, which is consistent with results on warm-start approaches for other types of variational algorithms~\cite{Niu_2023,Truger_2024}.

The improvement regarding time complexity over classical policy iteration depends on the speed-up gained by using the variational \gls{lse} solver, which has empirically shown to be super-polynomial under certain conditions~\cite{Bravo_2023} -- further investigated in~\cref{app:sparsity,app:condition_number}. Furthermore, it is possible to implement an $N \times N$ system using only $\log_2{N}$ qubits, providing an exponential improvement w.r.t.\ space complexity.

\section{\label{sec:experiments}Experimental End-to-End Realization}

\begin{table}[t]
    \centering
    \begin{tabular}{@{}c||cc|cc@{}}
        \toprule
        \multirow{2}{*}{~Iteration} & \multicolumn{2}{c|}{\gls{wsvarqpi}} & \multicolumn{2}{c}{\gls{varqpi}} \\
         & Average Steps & Percentage & Average Steps & Percentage~ \\
        \midrule
        1 & $1337\pm275$ & $100\%$ & $1337\pm275$ & $100\%$  \\
        2 & $1027\pm258$ & $100\%$ & $1313\pm295$ & $100\%$  \\
        3 & $\phantom{0}918\pm262$ & $\phantom{0}99\%$ & $1497\pm291$ & $\phantom{0}99\%$  \\
        4 & $\phantom{0}690\pm297$ & $\phantom{0}71\%$ & $1542\pm258$ & $\phantom{0}70\%$  \\
        5 & $\phantom{0}540\pm208$ & $\phantom{0}29\%$ & $1488\pm223$ & $\phantom{0}24\%$  \\
        6 & $\phantom{0}295\pm\phantom{0}40$ & $\phantom{00}8\%$ & $1578\pm158$ & $\phantom{00}6\%$  \\
        \bottomrule
    \end{tabular}
    \vspace{0.15cm}
    \captionsetup{name=Tab.,format=hang,indention=-0.9cm,font=footnotesize}
    \caption{Comparison of required steps for solving the \gls{lse} induced by \texttt{FrozenLake} (stochasticity $\beta=0.1$) with warm-start and random parameter initialization. The values are reported on a per-iteration basis, and averaged over $100$ random initial policies. Additionally we report, which percentage of runs has not converged after the indicated iteration number.}\label{tab:warm_start}%
\end{table}
\begin{figure}[t]
    \usetikzlibrary{calc}
    \centering
    \begin{tikzpicture}
        \centering
        \begin{axis}[
            name=plot1,
            xlabel=$\textbf{steps}$,
            ylabel=$\textbf{loss}$,
            xtick={0, 2000, 4000},
            ymode=log,
            xmin=0.0,xmax=4500.0,
            ymin=0.00005,ymax=2,
            label style={font=\footnotesize},
            tick label style={font=\footnotesize},
            ymajorgrids=true,
            axis x line=bottom, axis y line=left,
            width=0.5\linewidth,
            height=4cm,
            legend style={/tikz/every even column/.append style={column sep=0.1cm, row sep=0.1cm},at={(0.33,0.88)},anchor=east,yshift=-5mm,font=\scriptsize}]
            ]

            \addplot[line width=.7pt,solid,color=col_warm] %
            	table[x=step,y=loss,col sep=comma]{figures/data/warm_start.csv};

            \draw[color=black, densely dotted, line width=0.7pt](axis cs: 1378, 0.00001) -- (axis cs: 1378, 1.0);
            \draw[color=black, densely dotted, line width=0.7pt](axis cs: 2659, 0.00001) -- (axis cs: 2659, 1.0);
            \draw[color=black, densely dotted, line width=0.7pt](axis cs: 3681, 0.00001) -- (axis cs: 3681, 1.0);

        \end{axis}
        \begin{axis}[
            name=plot2,
            xlabel=$\textbf{steps}$,
            at={($(plot1.east)+(0.5cm,0)$)}, anchor=west,
            yticklabels={},
            xtick={0, 2000, 4000, 6000},
            ymode=log,
            xmin=0.0,xmax=6500.0,
            ymin=0.00005,ymax=2,
            label style={font=\footnotesize},
            tick label style={font=\footnotesize},
            ymajorgrids=true,
            axis x line=bottom, axis y line=left,
            width=0.64\linewidth,
            height=4cm,
            legend style={/tikz/every even column/.append style={column sep=0.1cm, row sep=0.1cm},at={(0.33,0.88)},anchor=east,yshift=-5mm,font=\scriptsize}]
            ]

            \addplot[line width=.7pt,solid,color=col_cold] %
            	table[x=step,y=loss,col sep=comma]{figures/data/cold_start.csv};

            \draw[color=black, densely dotted, line width=0.7pt](axis cs: 1378, 0.00001) -- (axis cs: 1378, 1.0);
            \draw[color=black, densely dotted, line width=0.7pt](axis cs: 3239, 0.00001) -- (axis cs: 3239, 1.0);
            \draw[color=black, densely dotted, line width=0.7pt](axis cs: 4728, 0.00001) -- (axis cs: 4728, 1.0);

        \end{axis}
       
        \end{tikzpicture}
    \captionsetup{font=footnotesize}
    \caption{\label{fig:warm_start}Loss curves for one instance of \gls{wsvarqpi} (left plot) and randomly initialized \gls{varqpi} (right plot) on \texttt{FrozenLake} with $\beta=0.1$. The vertical dotted lines indicate one iteration of policy evaluation with loss threshold $0.0001$. The procedure continues with the updated \gls{lse} after policy improvement -- with the previously optimal parameters, or the initial parameters from the first iteration. The large variance in the training procedure can be attributed to the use of a static learning rate of $0.01$. While learning rate schedulers might be used to reduce the fluctuations, the required additional hyperparameter tuning was found to be non-trivial for the proof-of-concept realization.}
\end{figure}


The implementation of the proposed \gls{varqpi} and \gls{wsvarqpi} algorithms described in~\cref{sec:method} use the \texttt{PennyLane}-based \texttt{variational-lse-solver} library~\cite{Meyer_2024}. For depth $d$, the parameterized ansatz $\ket{x(\bm{\alpha})}$ consists of $d$ layers of parameterized $U_3$ rotations, interleaved by nearest-neighbor CNOT layers, as depicted in~\cref{fig:vqpi} for depth $2$. In all experiments we used a learning rate of $0.01$ for the Adam optimizer~\cite{Kingma_2014} and discount factor $\gamma=0.9$. These hyperparameters were tuned in order to produce robust convergence for different initial conditions. 

For the experiments in this section we resorted to the \texttt{FrozenLake} environment~\cite{Brockman_2016}, using the standard \texttt{4x4} setup shown in~\cref{fig:vqpi}. We introduce stochasticity to the environment dynamics by instead transitioning to perpendicular states with probability $0 \leq \beta \leq \frac{1}{3}$. Comparable setups have frequently been considered through the lens of function approximation-based \gls{qrl}~\cite{Skolik_2022,Dragan_2022}. The $16$-state and $4$-action environment induces a $64 \times 64$-dimensional system matrix of the underlying \gls{lse} following~\cref{eq:bellman_equation_lse}. With an ancilla qubit for evaluating the loss function with the Hadamard test~\cite{Bravo_2023}, the realization requires $7$ qubits in total. An alternative formulation based on the Hadamard-overlap test~\cite{Bravo_2023} requires $13$ qubits, but significantly reduces the number of multi-qubit operations. While this is well within the reach of existing \gls{nisq} hardware, we resort to classical simulation, due to the limited access to these systems. For the variational quantum policy evaluation routine we constructively decompose the real-valued system matrix $A_{\pi}$ into a linear combination of $4$ orthogonal matrices~\cite{Zhan_2001,Li_2002}, satisfying~\cref{eq:unitary_decomposition}. Details on this procedure and a discussion of potential alternatives are deferred to~\cref{app:decomposition}.

In the following, we report the proof-of-principle realization of the proposed routines, starting out with a comparison of the efficiency of both variants in~\cref{subsec:warm_start}. The scalability of the approach is demonstrated in~\cref{subsec:upscaling}. An analysis of the resource requirements is deferred to~\cref{sec:resource_analysis}.


\subsection{\label{subsec:warm_start}Warm-Start Parameter Initialization}

We compare the performance of the \gls{varqpi} algorithm on the \texttt{FrozenLake} environment (random parameter initialization in each iteration) to \gls{wsvarqpi} (parameter initialization after first iteration as described in~\cref{subsec:warm_start}). The circuit depth is $d=12$, with early-stopping after reaching a loss threshold of $0.0001$ -- those hyperparameters will be justified in the following section. Averaged over $100$ random initial policies for a $\beta=0.1$-stochastic \texttt{FrozenLake} environment, the warm-start approach converges within $4.1 \pm 0.9$ iterations, requiring $3943 \pm 952$ training steps. The randomly-initialized variant requires $4.0 \pm 0.9$ iterations with $5663 \pm 1366$ steps. While the negligible difference in iterations can be attributed to statistical variance, the reduction w.r.t.\ required updates by approx. $30\%$ is clearly significant. From~\cref{tab:warm_start} it becomes apparent that this difference is more significant in the later iterations. This can be explained by the resulting greedy policy only changing slightly between two successive cycles of policy evaluation and improvement, enhancing the impact of \gls{wsvarqpi}. One might expect an even clearer reduction in overall training complexity for larger and more complex environments, where more iterations are required for convergence. We also observed similar results for setups with lower and higher stochasticity.

In \cref{fig:warm_start} we explicitly report the losses observed during one instance of \gls{wsvarqpi} compared to randomly-initialized \gls{varqpi}. The warm-start version indeed seems to get initialized to a more desirable part of the optimization landscape, indicated also by the successively reduced initial loss value after policy improvement. For both setups it can be observed, that the loss values start to fluctuate starting from a loss value of about $0.01$ -- equivalent to the learning rate. This variance can be reduced by using learning rate schedulers, with initial experiments suggesting a linear decay to be best suited. However, this introduces additional hyperparameters, which impede a fair comparison of both routines. While the fluctuations did not hinder the convergence in the considered setups, statements regarding scalability require further investigation.

\subsection{\label{subsec:upscaling}Up-Scaling to Larger Environment}

To strengthen the potential of \gls{wsvarqpi}, we execute the algorithm on the $\beta=0.1$-stochastic instance of \texttt{FrozenLake8x8} in~\cref{fig:largeFL}. This $64$-state and $4$-action setup enlarges the \gls{lse} system matrix to a size of $256 \times 256$. This constitutes a $16$-fold increase of the matrix size compared to the \texttt{4x4} setup. However, the spatial resources for the quantum-enhanced realization only slightly increase from $7$ to $9$ qubits for the Hadamard formulation, and from $13$ to $17$ qubits for the Hadamard-overlap formulation.

To additionally account for the larger environment complexity, we increase the circuit depth to $d=24$, while keeping the loss threshold at $0.0001$. Note, however, that this only provides an upper bound on the required resources, and should not be used to infer scaling behavior regarding circuit depth. Within $9$ iterations and overall $82160$ training steps the learned greedy policy solves the environment. This result was re-produced for $10$ random initial policies with only slight deviations in iteration and step number count. While this result does not proof general scalability, it certainly should motivate extending and refining \gls{varqpi} and \gls{wsvarqpi}.

\begin{figure}
    \centering
    \includegraphics[width=0.25\textwidth]{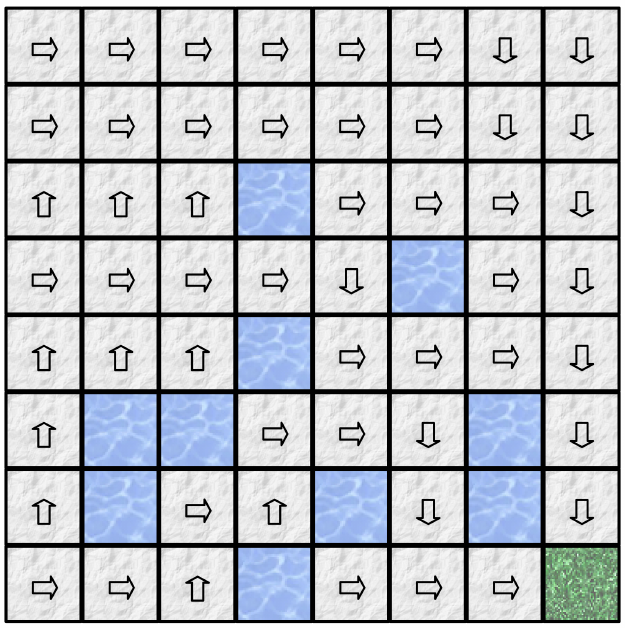}
    \captionsetup{font=footnotesize}
    \caption{The policy found for the \texttt{FrozenLake8x8} environment with $\beta=0.1$-stochasticity using \gls{wsvarqpi}. Training required $82160$ steps in $9$ policy evaluation and improvement cycles. The experiments use a depth $d=24$ ansatz and a loss threshold of $0.0001$. The solution perfectly aligns with the ground truth for this configuration.}
    \label{fig:largeFL}
\end{figure}

\section{\label{sec:resource_analysis}Resource Analysis of Quantum Subroutine}

The spatial resources -- i.e.\ the number of qubits -- for variational quantum policy iteration are determined by the size of the system matrix. More concretely, for size $N := \abs{\mathcal{A}} \times \abs{\mathcal{S}}$, the variational ansatz $\ket{x(\bm{\alpha})}$ has to act on $n := \lceil \log_2 N \rceil$ qubits. For the training procedure described in~\cref{subsec:lse} the overall system size scales as $n+1$ if using the Hadamard test, and as $2n+1$ if the loss function is evaluated with the Hadamard-overlap test~\cite{Bravo_2023}.

In the following, we investigate two other crucial components of the overall algorithmic scaling behavior: In~\cref{subsec:threshold} we demonstrate, that the threshold on the loss function in practice can be selected several orders of magnitude higher than suggested by~\cref{eq:threshold}, allowing for significantly faster convergence. In~\cref{subsec:depth} we analyze the depth of the variational ansatz required for reliably solving the linear system.

\begin{figure}[t]
    \usetikzlibrary{calc}
    \centering
    \begin{tikzpicture}
        \centering
        \begin{axis}[
            name=plot2,
            ylabel=$\textbf{success rate}$,
            xmode=log,
            xtick={0.00002, 0.00005, 0.0001, 0.0002, 0.0005, 0.001, 0.002, 0.005, 0.01, 0.02, 0.05},
            xticklabels={$2$, $5$, $1$, $2$, $5$, $1$, $2$, $5$, $1$, $2$, $5$},
            ytick={0.0, 0.5, 1.0},
            yticklabels={$0\%$,$50\%$,$100\%$},
            xmin=0.000016,xmax=0.065,
            ymin=0.0,ymax=1.1,
            clip=false,
            label style={font=\footnotesize},
            tick label style={font=\footnotesize},
            grid=both,
            axis x line=bottom, axis y line=left,
            width=0.99\linewidth,
            height=4cm,
            legend style={/tikz/every even column/.append style={column sep=0.1cm, row sep=0.1cm},at={(0.3,0.6)},anchor=east,yshift=-5mm,font=\scriptsize}]
            ]

            \addplot[line width=.7pt,solid,color=col_per00,mark=o] %
            	table[x=threshold,y=per0.0,col sep=comma]{figures/data/threshold.csv};

            \addplot[line width=.7pt,solid,color=col_per01,mark=o] %
            	table[x=threshold,y=per0.1,col sep=comma]{figures/data/threshold.csv};

            \addplot[line width=.7pt,solid,color=col_per03,mark=o] %
            	table[x=threshold,y=per0.3,col sep=comma]{figures/data/threshold.csv};

            \legend{$\beta=0.0$,$\beta=0.1$,$\beta=0.3$};

            \draw[decorate,decoration={brace,mirror}]
                ([yshift=-15pt]axis cs:0.000018,0.0) --
                node[below=2pt] {\footnotesize$\cdot10^{-5}$} 
                ([yshift=-15pt]axis cs:0.000055,0.0);

            \draw[decorate,decoration={brace,mirror}]
                ([yshift=-15pt]axis cs:0.00009,0.0) --
                node[below=2pt] {\footnotesize$\cdot10^{-4}$} 
                ([yshift=-15pt]axis cs:0.00055,0.0);

            \draw[decorate,decoration={brace,mirror}]
                ([yshift=-15pt]axis cs:0.0009,0.0) --
                node[below=2pt] {\footnotesize$\cdot10^{-3}$} 
                ([yshift=-15pt]axis cs:0.0055,0.0);

            \draw[decorate,decoration={brace,mirror}]
                ([yshift=-15pt]axis cs:0.009,0.0) --
                node[below=2pt] {\footnotesize$\cdot10^{-2}$} 
                ([yshift=-15pt]axis cs:0.055,0.0);


        \end{axis}
        \begin{axis}[
            name=plot3,
            at={($(plot2.south)+(0,-1.25cm)$)}, anchor=north,
            xlabel=$\textbf{loss threshold}$,
            xlabel shift=18pt,
            ylabel=$\textbf{steps}$,
            xmode=log,
            xtick={0.00002, 0.00005, 0.0001, 0.0002, 0.0005, 0.001, 0.002, 0.005, 0.01, 0.02, 0.05},
            xticklabels={$2$, $5$, $1$, $2$, $5$, $1$, $2$, $5$, $1$, $2$, $5$},
            ytick={0, 1500, 3000},
            yticklabels={$0$,$1500$, $3000$},
            xmin=0.000016,xmax=0.065,
            ymin=0.0,ymax=3750.0,
            clip=false,
            label style={font=\footnotesize},
            tick label style={font=\footnotesize},
            grid=both,
            axis x line=bottom, axis y line=left,
            width=0.99\linewidth,
            height=3.25cm,
            legend style={/tikz/every even column/.append style={column sep=0.1cm, row sep=0.1cm},at={(0.33,0.88)},anchor=east,yshift=-5mm,font=\scriptsize}]
            ]

            \addplot+[draw=none,name path=A,no markers] %
            	table[x=threshold,y=per0.0_plus,col sep=comma]{figures/data/threshold_depth.csv};
            \addplot+[draw=none,name path=B,no markers] %
            	table[x=threshold,y=per0.0_minus,col sep=comma]{figures/data/threshold_depth.csv};
            \addplot[col_per00!40] fill between[of=A and B];
            \addplot[line width=.7pt,solid,color=col_per00,mark=o] %
            	table[x=threshold,y=per0.0,col sep=comma]{figures/data/threshold_depth.csv};

            \addplot+[draw=none,name path=A,no markers] %
            	table[x=threshold,y=per0.1_plus,col sep=comma]{figures/data/threshold_depth.csv};
            \addplot+[draw=none,name path=B,no markers] %
            	table[x=threshold,y=per0.1_minus,col sep=comma]{figures/data/threshold_depth.csv};
            \addplot[col_per01!40] fill between[of=A and B];
            \addplot[line width=.7pt,solid,color=col_per01,mark=o] %
            	table[x=threshold,y=per0.1,col sep=comma]{figures/data/threshold_depth.csv};

            \addplot+[draw=none,name path=A,no markers] %
            	table[x=threshold,y=per0.3_plus,col sep=comma]{figures/data/threshold_depth.csv};
            \addplot+[draw=none,name path=B,no markers] %
            	table[x=threshold,y=per0.3_minus,col sep=comma]{figures/data/threshold_depth.csv};
            \addplot[col_per03!40] fill between[of=A and B];
            \addplot[line width=.7pt,solid,color=col_per03,mark=o] %
            	table[x=threshold,y=per0.3,col sep=comma]{figures/data/threshold_depth.csv};

            \draw[decorate,decoration={brace,mirror}]
                ([yshift=-15pt]axis cs:0.000018,0.0) --
                node[below=2pt] {\footnotesize$\cdot10^{-5}$} 
                ([yshift=-15pt]axis cs:0.000055,0.0);

            \draw[decorate,decoration={brace,mirror}]
                ([yshift=-15pt]axis cs:0.00009,0.0) --
                node[below=2pt] {\footnotesize$\cdot10^{-4}$} 
                ([yshift=-15pt]axis cs:0.00055,0.0);

            \draw[decorate,decoration={brace,mirror}]
                ([yshift=-15pt]axis cs:0.0009,0.0) --
                node[below=2pt] {\footnotesize$\cdot10^{-3}$} 
                ([yshift=-15pt]axis cs:0.0055,0.0);

            \draw[decorate,decoration={brace,mirror}]
                ([yshift=-15pt]axis cs:0.009,0.0) --
                node[below=2pt] {\footnotesize$\cdot10^{-2}$} 
                ([yshift=-15pt]axis cs:0.055,0.0);

        \end{axis}
    
        \end{tikzpicture}
    \captionsetup{font=footnotesize}
    \caption{\label{fig:threshold}Success rates of variational quantum policy evaluation with a depth $12$ circuit for different loss thresholds and therefore required training steps. The results are averaged over $1000$ random instances of \texttt{FrozenLake} with varying stochasticity $\beta$ -- the bands in the lower plot denoting standard deviations. Training is terminated once the loss value $C_{G}$ has declined below the given threshold. The success rate denotes the percentage of runs for which \cref{eq:greedy_action_selection} agrees with the ground truth. To account for sampling noise, we tolerate a deviation of $0.001$ in the associated $Q$-values.}
\end{figure}

\subsection{\label{subsec:threshold}Loss Threshold and Environment Stochasticity}

As the prepared quantum state $\ket{x(\bm{\alpha})}$ is used for greedy policy improvement following~\cref{eq:greedy_policy}, we do not need a perfectly accurate solution to the underlying \gls{lse}. It is enough, if the resulting greedy policy aligns with the ground truth -- on all non-terminal states. The operational meaning of the loss function from~\cref{eq:threshold} is proportional to the desired deviation $\epsilon$, and inverse proportional to the condition number $\kappa$~\cite{Bravo_2023}.

We select the tolerated deviation to be $0.001$, motivated by the number of shots for approximating $M(s,a)$ being $1000$. For the deterministic 
\texttt{FrozenLake} environment the system matrix for a random initial policy has a condition number of $85.6$, while the $\beta=0.3$-stochastic version exhibits a condition number of $67.4$. It is demonstrated in~\cref{app:condition_number}, that the condition numbers for more generic \gls{rl} environments are in a similar range. Following~\cref{eq:threshold}, the loss thresholds are $C_G \geq \frac{0.001^2}{85.6^2} \approx 1.36 \cdot 10^{-12}$ and $C_G \geq \frac{0.001^2}{67.4^2} \approx 2.20 \cdot 10^{-12}$, respectively. In principle, this this should allow for a comparable accuracy with an about $0.6$ times lower threshold bound for the stochastic version. However, in practice the deterministic environment is easier to solve, as we will see in the following.


In \cref{fig:threshold} we demonstrate, that those bounds are pretty loose for quantum policy iteration, which converges w.r.t.\ the $\ell_{\infty}$ norm. One cycle of quantum policy evaluation and successive policy improvement is successful, if the greedy policy aligns with the ground truth on all non-terminal states. We interpret the overall algorithmic setup to be reliable, if the condition holds for $95\%$ -- i.e.\ two standard deviations -- of all random initial policies. Variational quantum policy evaluation works reliably in deterministic \texttt{FrozenLake} environments for loss thresholds up to $0.05$, improving upon the theoretical bound established above by $11$ orders of magnitude. For stochastic environments a threshold value of about $0.0001$ is required for stable results, still constituting an improvement by about $8$ orders of magnitude. It is currently unclear if these practical improvements transfer to the results from ~\cref{app:condition_number}, where we considered the general scalability of quantum-enhanced \gls{lse} solvers for policy iteration. As indicated in the lower part of~\cref{fig:threshold}, the smaller loss thresholds required for the stochastic environments are partially compensated by a faster convergence compared to deterministic setups.

\subsection{\label{subsec:depth}Depth of Variational Ansatz}

The depth of the involved circuits mainly depends on two factors: The variational ansatz $\ket{x(\bm{\alpha})}$ and the explicit realization of the unitary decomposition following~\cref{eq:unitary_decomposition}. The latter one can be improved by developing sophisticated matrix decomposition and unitary synthesis tools, which is not part of this work, but is superficially discussed in~\cref{app:decomposition}. The former one will be analyzed in the following.

\begin{figure}[t]
    \pgfplotsset{scaled y ticks=false}
    \usetikzlibrary{calc}
    \centering
    \begin{tikzpicture}
        \centering
        \begin{axis}[
            name=plot2,
            ylabel=$\textbf{success rate}$,
            xtick={2,4,6,8,10,12,14},
            ytick={0.0, 0.5, 1.0},
            yticklabels={$0\%$,$50\%$,$100\%$},
            xmin=1.3,xmax=14.2,
            ymin=0.0,ymax=1.1,
            label style={font=\footnotesize},
            tick label style={font=\footnotesize},
            grid=both,
            axis x line=bottom, axis y line=left,
            width=0.99\linewidth,
            height=4cm,
            legend style={/tikz/every even column/.append style={column sep=0.1cm, row sep=0.1cm},at={(0.96,0.67)},anchor=east,yshift=-5mm,font=\scriptsize}]
            ]

            \addplot[line width=.7pt,solid,color=col_per00,mark=o] %
            	table[x=depth,y=validated,col sep=comma]{figures/data/resources_per=0_0.csv};

            \addplot[line width=.7pt,solid,color=col_per01,mark=o] %
            	table[x=depth,y=validated,col sep=comma]{figures/data/resources_per=0_1.csv};


            \legend{$\beta=0.0$,$\beta=0.1$};


        \end{axis}
        \begin{axis}[
            name=plot3,
            at={($(plot2.south)+(0,-0.75cm)$)}, anchor=north,
            xlabel=$\textbf{depth}$,
            xtick={2,4,6,8,10,12,14},
            ylabel=$\textbf{steps}$,
            ytick={0, 4000, 8000},
            yticklabels={$0$,$4000$,$8000$},
            xmin=1.3,xmax=14.2,
            ymin=0.0,ymax=12000.0,
            label style={font=\footnotesize},
            tick label style={font=\footnotesize},
            grid=both,
            axis x line=bottom, axis y line=left,
            width=0.99\linewidth,
            height=4cm,
            legend style={/tikz/every even column/.append style={column sep=0.1cm, row sep=0.1cm},at={(0.33,0.88)},anchor=east,yshift=-5mm,font=\scriptsize}]
            ]

            \addplot+[draw=none,name path=A,no markers] %
            	table[x=depth,y=raw_steps_plus,col sep=comma]{figures/data/resources_per=0_0.csv};
            \addplot+[draw=none,name path=B,no markers] %
            	table[x=depth,y=raw_steps_minus,col sep=comma]{figures/data/resources_per=0_0.csv};
            \addplot[col_per00!40] fill between[of=A and B];
            \addplot[line width=.7pt,solid,color=col_per00,mark=o] %
            	table[x=depth,y=raw_steps,col sep=comma]{figures/data/resources_per=0_0.csv};

            \node at (axis cs:2.07,1800)[color=col_per00,anchor=center] {\scriptsize $49.1\%$};
            \node at (axis cs:4.07,900)[color=col_per00,anchor=center] {\scriptsize $70.7\%$};
            \node at (axis cs:6.07,4200)[color=col_per00,anchor=center] {\scriptsize $29.5\%$};
            \node at (axis cs:8.07,11200)[color=col_per00,anchor=center] {\scriptsize $13.9\%$};

            \addplot+[draw=none,name path=A,no markers] %
            	table[x=depth,y=raw_steps_plus,col sep=comma]{figures/data/resources_per=0_1.csv};
            \addplot+[draw=none,name path=B,no markers] %
            	table[x=depth,y=raw_steps_minus,col sep=comma]{figures/data/resources_per=0_1.csv};
            \addplot[col_per01!40] fill between[of=A and B];
            \addplot[line width=.7pt,solid,color=col_per01,mark=o] %
            	table[x=depth,y=raw_steps,col sep=comma]{figures/data/resources_per=0_1.csv};

            \node at (axis cs:2.07,11400)[color=col_per01,anchor=center] {\scriptsize $11.7\%$};
            \node at (axis cs:4.07,10700)[color=col_per01,anchor=center] {\scriptsize $30.7\%$};
            \node at (axis cs:6.07,11100)[color=col_per01,anchor=center] {\scriptsize $23.1\%$};
            \node at (axis cs:8.07,6200)[color=col_per01,anchor=center] {\scriptsize $33.0\%$};

        \end{axis}
    
        \end{tikzpicture}
    \captionsetup{font=footnotesize}
    \caption{\label{fig:resources}Success rates of variational quantum policy evaluation for different circuit depths and therefore required training steps. The results are averaged over $1000$ random instances of \texttt{FrozenLake} with varying stochasticity $\beta$ -- the bands in the lower plot denoting standard deviations. The success rates denotes the percentage of runs for which \cref{eq:greedy_action_selection} agrees with the ground truth. To account for sampling noise, we tolerate a deviation of $0.001$ in the associated $Q$-values. The percentages in the lower plot denote the ration of runs below $100\%$ that achieved the loss threshold $0.0001$ with less then $10000$ steps, after which training is aborted.}
\end{figure}

As before, we select the variational circuit ansatz on $n$ qubits with $d \geq 1$ layers as 
\begin{align}
    \ket{x(\bm{\alpha})} &= \underbrace{U_{\bm{\alpha}_{d-1}} U_{\mathrm{CX}}}_{\text{layer } d-1} \cdots \underbrace{U_{\bm{\alpha}_1} U_{\mathrm{CX}}}_{\text{layer } 1} \underbrace{U_{\bm{\alpha}_0}}_{\text{layer } 0}, \\
    \text{with} &\phantom{=} U_{\bm{\alpha}_i} = U_3(\bm{\alpha}_{i,0}) \otimes \cdots \otimes U_3(\bm{\alpha}_{i,n-1}), \\
    &\phantom{=} U_{\mathrm{CX}} = \mathrm{CX}_{0,1} \cdots \mathrm{CX}_{n-2,n-1} \mathrm{CX}_{n-1,0},
\end{align}
as also depicted in~\cref{fig:vqpi}. In order to determine the actually required number of layers, we report experiments for a varying depth in~\cref{fig:resources}. Again it becomes apparent that the deterministic \texttt{FrozenLake} environment seems to be easier to solve, as for a given depth the probability that the approximated and ground truth greedy actions align is higher. In general, a more expressive circuit allows for a higher success rate, with reliable accuracy achieved from depth $d=10$ upwards. Interestingly, this does not consistently hold true for the convergence rate. This measure quantifies the likelihood that the respective ansatz actually allows convergence to the loss threshold $0.0001$ in $10000$ steps. As indicated in the lower plot, for setups with $8$ qubits or less there is no consistent pattern. This can be interpreted as a motivation to investigate ans\"atze that exploit some structure of the \gls{lse} to reduce resource requirements. For deeper circuits this fluctuation disappears with all instances converging. As the required steps for convergence are reduced with increasing circuit depth, we identified $d=12$ to provide a good balance regarding resource usage.



\glsresetall

\section{\label{sec:conclusion}Conclusion}

In this work we introduced \gls{varqpi}, an algorithm combining quantum-enhanced policy evaluation with classical policy improvement. It uses a variational LSE solver to perform policy evaluation potentially more efficiently than possible with classical solution methods. The greedy policy improvement step is based on efficient $\ell_{\infty}$-tomography, avoiding the typical bottleneck of exponential sampling complexity for retrieving quantum states.

We demonstrated that the system matrices induced by policy iteration are compatible with requirements imposed by the quantum sub-routines. Specifically, we showed that generic environments satisfy strict sparsity conditions and the associated LSEs are well-conditioned. Both results support the scalability of the approach, which is eminent for potential quantum advantage. One existing bottleneck is the unitary decomposition and synthesis of the matrices, which is left for future work.

The efficiency of the routine was empirically demonstrated on different configurations of the \texttt{FrozenLake} environment. We discovered that a version with warm-start parameter initialization (WS-VarQPI) can significantly improve the convergence behavior. We conducted an analysis of different hyperparameters settings to support the robustness of the approach.

We are confident that (WS-)\gls{varqpi} opens the way for a promising line of research on \gls{nisq}-compatible quantum reinforcement learning. The successful application to a complex \texttt{FrozenLake8x8} environment underlines the potential of the approach. Future research could include the extension to model-free scenarios and validation experiments on quantum hardware.

\section*{Code Availability}

An implementation of \gls{varqpi} and WS-VarQPI, including the considered environments, is available in the repository \url{https://github.com/nicomeyer96/ws-varqpi}. Usage details and instructions on how to reproduce the main results of this paper can be found on the \texttt{README} file. Further information and data is available upon reasonable request.


\appendix

Both fault-tolerant and \gls{nisq}-compatible \gls{lse} solvers promise speed-ups only under certain restrictions. Two conditions are sparsity and low condition number of the system matrix $A$~\cite{Childs_2017,Cherrat_2023}. While the former is not explicitly stated in the error bounds of the variational \gls{lse} solver, the scaling and matrix decomposition results implicitly assume sparsity~\cite{Bravo_2023}. We demonstrate in the following, that the linear systems associated with typical \gls{rl} setups are well-behaved regarding these conditions.

\subsection{\label{app:sparsity}Sparsity for Local Dynamics}

Quantum linear system solvers assume sparsity of the $N \times N$ system matrix $A$, with only $\mathcal{O}(N \cdot \mathrm{polylog}(N))$ non-zero entries~\cite{Childs_2017,Bravo_2023}. A stronger assumption is that each of the $N$ rows of $A$ has at most $\mathcal{O}(\mathrm{polylog}(N)))$ non-zero elements.

For state set $\mathcal{S}$, action set $\mathcal{A}$, and dynamics $p$, we define the set of states that can be reached from $s$ under action $a$ as
\begin{align}
    \label{eq:local_neighborhood}
    \mathcal{S}_{s,a} = \left\{ s' \mid p(s'|s,a) > 0 \right\},
\end{align}
and the set of states that can transition to $s'$ under action $a$ as
\begin{align}
    \label{eq:local_neighborhood_reverse}
    \bar{\mathcal{S}}_{s',a} = \left\{ s \mid p(s'|s,a) > 0 \right\}.
\end{align}

\begin{definition}[local dynamics]
    \label{def:local_dynamics}
    An environment with state set $\mathcal{S}$ and action set $\mathcal{A}$ has local dynamics, if it holds
    \begin{align}
        \abs{\mathcal{S}_{s,a}} &\leq \log_2(N), \\ \abs{\bar{\mathcal{S}}_{s',a}} &\leq \log_2(N)
    \end{align}
    for all $s, s' \in \mathcal{S}$, $a \in \mathcal{A}$, with $N:=\abs{\mathcal{S}}$.
\end{definition}
\noindent These assumptions are reasonable in the context of \gls{rl}, as transition probabilities typically decrease with (spatial) distance between states. As in \gls{qpi} the policy is deterministic, the policy matrix defined in~\cref{eq:bellman_equation_matrix} has exactly one non-zero entry in each row. Consequently, $P_{\pi} = P \Pi$ is row-stochastic with at most $\log_2(N)$ non-zero entries per row and column for~\cref{def:local_dynamics}. The system matrix $A_{\pi} = \mathbb{I} - \gamma P_{\pi}$ contains at most one additional non-zero element in each row, not changing the overall scaling. Therefore, locally-dynamic environments satisfy the sparsity condition posed by quantum-enhanced \gls{lse} solvers stated above. Additionally, this guaranteed sparsity is a desirable property for allowing an efficient unitary decomposition of $A_{\pi}$, as further elaborated in~\cref{app:decomposition}.

Depending on the concrete dynamics function $p$, one can consider special cases of environments satisfying~\cref{def:local_dynamics}:

\begin{definition}[deterministic dynamics]
    \label{def:deterministic_dynamics}
    An environment with state set $\mathcal{S}$ and action set $\mathcal{A}$ has deterministic dynamics, if
    \begin{align}
        \abs{\mathcal{S}_{s,a}} &= 1, \\
        \abs{\bar{\mathcal{S}}_{s',a}} &\leq \log_2(N)
    \end{align}
    for all $s,s' \in \mathcal{S}$, and $a \in \mathcal{A}$. Consequently, for all $s \in \mathcal{S}$, $a \in \mathcal{A}$ there is only one $s' \in \mathcal{S}$ such that
    \begin{equation}
        p(s'|s,a) = 1,
    \end{equation}
    while for all other $s' \in \mathcal{S}$ it holds $p(s'|s,a) = 0$.
\end{definition}

\begin{definition}[uniform local dynamics]
    \label{def:uniform_local_dynamics}
    An environment with state set $\mathcal{S}$ and action set $\mathcal{A}$ has uniform local dynamics, if
    \begin{align}
        \abs{\mathcal{S}_{s,a}} = \log_2(N), \\
        \abs{\bar{\mathcal{S}}_{s',a}} \leq \log_2(N)
    \end{align}
    for all $s, s' \in \mathcal{S}$, $a \in \mathcal{A}$, with $N = \abs{\mathcal{S}}$. Furthermore, for all $s \in \mathcal{S}$, $a \in \mathcal{A}$ there are $\log_2(N)$ different $s' \in \mathcal{S}$ such that 
    \begin{equation}
        p(s'|s,a) = \frac{1}{\log_2(N)},
    \end{equation}
    while for all other $s' \in \mathcal{S}$ it holds $p(s'|s,a) = 0$.
\end{definition}

\begin{definition}[exponential local dynamics]
    \label{def:exponential_local_dynamics}
    An environment with state set $\mathcal{S}$ and action set $\mathcal{A}$ has exponential local dynamics, if
    \begin{align}
        \abs{\mathcal{S}_{s,a}} &= \log_2(N), \\
        \abs{\bar{\mathcal{S}}_{s',a}} &\leq \log_2(N)
    \end{align}
    for all $s \in \mathcal{S}$, $a \in \mathcal{A}$, with $N = \abs{\mathcal{S}}$. Furthermore, for all $s \in \mathcal{S}$, $a \in \mathcal{A}$ there are $\log_2(N)$ different $s' \in \mathcal{S}$ such that 
    \begin{equation}
        p(s'|s,a) = \frac{\exp \left(\left( 1 - \nicefrac{1}{\beta} \right) d(s,s') \right)}{\sum_{s'' \in \mathcal{S}_{s,a}}  \exp \left( \left( 1 - \nicefrac{1}{\beta} \right) d(s,s'') \right)},
    \end{equation}
    with perturbation parameter $\beta \in (0,1]$ and some distance measure $d: \mathcal{S} \times \mathcal{S} \mapsto \mathbb{R}_+$, while for all other $s' \in \mathcal{S}$ it holds $p(s'|s,a) = 0$.
\end{definition}

\noindent The \emph{exponential local dynamics} in~\cref{def:exponential_local_dynamics} assume, that the transition probabilities vanish exponentially within a logarithmic local neighborhood. The \emph{deterministic dynamics} in~\cref{def:deterministic_dynamics} are a special case with $\beta \to 0$. Furthermore, the \emph{uniform local dynamics} from~\cref{def:uniform_local_dynamics} describe the opposite extreme with $\beta = 1$. While all environments with dynamics following~\cref{def:local_dynamics,def:deterministic_dynamics,def:uniform_local_dynamics,def:exponential_local_dynamics} satisfy the imposed sparsity constraints, guarantees w.r.t.\ the condition number vary, as will be discussed in the following.

\subsection{\label{app:condition_number}Bounds on Condition Number}
Quantum-enhanced \gls{lse} solvers require a small condition number $\kappa = \frac{\sigma_{\mathrm{max}}}{\sigma_{\mathrm{min}}}$ of the system matrix $A_{\pi}$~\cite{Childs_2017,Cherrat_2023}, where $\sigma_{\mathrm{max}}$ and $\sigma_{\mathrm{min}}$ denote the largest and smallest singular values, respectively. This can be seen in the threshold bounds from~\cref{eq:threshold}, where $\kappa$ enters in the denominator. We demonstrate that the \gls{lse} associated with typical \gls{rl} setups are well-conditioned, supporting the feasibility of larger-scale \gls{qpi}. 

We need to mention, that the analytical bound to the condition number derived in Cherrat et al.~\cite{Cherrat_2023} does not hold for generic \gls{rl} setups. The authors state, that for $A_{\pi} := \mathbb{I} - \gamma P_{\pi}$ the condition number $\kappa(A_{\pi})$ is upper-bounded by $\frac{1+\gamma}{1-\gamma}$ (see Theorem 7 of~\cite{Cherrat_2023}). However, the proof assumes that $\| P_{\pi} \| = 1$, which only holds for doubly-stochastic $P_{\pi}$~\cite{Horn_2012}. The matrices $P_{\pi}$ of typical \glspl{mdp} in general are only row-stochastic, not allowing this assumption. 

In the following, we derive corrected bounds for the general setup, empirically show that the systems are well-conditioned, and re-discover the bounds stated by Cherrat et al. for a special case. Therefore, we first state and prove some helpful properties of matrix norms:

\begin{proposition}
    \label{prop:row_norm_inverse}
    For a square matrix $A \in \mathbb{C}^{N \times N}$ it holds
    \begin{align}
        \| A^{-1} \|_{\infty} \leq \frac{1}{\min_i \lbrace \abs{A_{ii}} - \sum_{i \neq j} \abs{A_{ij}} \rbrace },
    \end{align}
    assuming the denominator on the right side is positive.
\end{proposition}
\begin{proof}
    Let us define $w = Av$, with column vectors $v, w \in \mathbb{C}^{N}$. By following the definition of an arbitrary operator norm as $\| A \| = \sup \lbrace \| Av \|: \| v \| = 1 \rbrace$ we need to show
    \begin{align}
        &\| w \|_{\infty} \leq 1 \Rightarrow \| v \|_{\infty} \leq \frac{1}{\min_i \lbrace \abs{A_{ii}} - \sum_{i \neq j} \abs{A_{ij}} \rbrace} \\
        \Leftrightarrow~ &\| w \|_{\infty} \leq \min\nolimits_i \lbrace \abs{A_{ii}} - \sum\nolimits_{i \neq j} \abs{A_{ij}} \rbrace \Rightarrow \| v \|_{\infty} \leq 1,
    \end{align}
    with a stronger statement reading as
    \begin{align}
        \label{eq:axis_aligned_condition}
        \forall i,j: ~\abs{w_i} \leq \abs{A_{ii}} - \sum\nolimits_{i \neq j} \abs{A_{ij}} \Rightarrow \forall i: ~\abs{v_i} \leq 1.
    \end{align}
    Let us assume the contrary, i.e.\ it exists an index $i$, such that $\abs{v_i} > 1$. Then it holds
    \begin{align}
        \abs{w_i} &= \left| A_{ii}v_i - \sum\nolimits_{i \neq j} A_{ij}(-v_j) \right| \\
        & \geq |A_{ii}||v_i| - \sum\nolimits_{i \neq j} |A_{ij}||v_j| \\
        & > |A_{ii}| - \sum\nolimits_{i \neq j} |A_{ik}|,
    \end{align}
    contradicting the condition from~\cref{eq:axis_aligned_condition}, and thereby proving the original statement by contraposition.
\end{proof}

\begin{proposition}
    \label{prop:col_norm_inverse}
    For a square matrix $A \in \mathbb{C}^{N \times N}$ it holds
    \begin{align}
        \| A^{-1} \|_{1} \leq \frac{1}{\min_i \lbrace \abs{A_{ii}} - \sum_{i \neq j} \abs{A_{ji}} \rbrace },
    \end{align}
    assuming the denominator on the right side is positive.
\end{proposition}
\begin{proof}
    Let us define $w = vA$, with row vectors $v, w \in \mathbb{C}^{N}$. The proof is equivalent to the proof of~\cref{prop:row_norm_inverse}.
\end{proof}
\noindent We prove an upper bound to the condition number of environments with general local dynamics:
\begin{theorem}
    \label{thm:general_local}
    Let $A_{\pi} := \mathbb{I} - \gamma P_{\pi}$ be the $N \times N$-dimensional system matrix of a \gls{lse} with underlying local dynamics following~\cref{def:local_dynamics}. The condition number of the matrix is upper bounded by
    \begin{align}
        \kappa(A_{\pi}) \leq \sqrt{N + \gamma \cdot N \log_2(N)} \cdot \frac{\sqrt{1 + \gamma}}{1 - \gamma}.
    \end{align}
\end{theorem}
\begin{proof}
    We first show an upper bound to the largest singular value, using H\"older's inequality:
    \begin{align}
        \sigma_{max}(A_{\pi}) &\leq \sqrt{\| A_{\pi} \|_{\infty}} \sqrt{\| A_{\pi} \|_{1}} \\
        &\leq \sqrt{1 + \gamma} \sqrt{1 + \gamma \cdot \log_2(N)}
    \end{align}
    The bound for the maximum row sum follows directly from the row-stochasticity of $P_{\pi}$, while the bound for the maximum column sum additionally employs the logarithmically restricted locality of the dynamics underlying $P_{\pi}$.

    \noindent Next, we show a lower bound to the smallest singular value, using the Cauchy-Schwarz inequality and~\cref{prop:row_norm_inverse}:
    \begin{align}
        \label{eq:estimation_sigma_min}
        \sigma_{\min}(A_{\pi}) = \frac{1}{\| A_{\pi}^{-1} \|}
        &\geq \frac{1}{\sqrt{N} \| A_{\pi}^{-1} \|_{\infty}} \\
        &\geq \frac{1}{\sqrt{N} \frac{1}{1 - \gamma}}
    \end{align} 

    \noindent Plugging these results into the definition of the condition number reveals the stated bound.
\end{proof}
\noindent Assuming environments with uniform local dynamics, it is possible to derive a much tighter bound:
\begin{theorem}
    \label{thm:uniform_local}
    Let $A_{\pi} := \mathbb{I} - \gamma P_{\pi}$ be the $N \times N$-dimensional system matrix of a \gls{lse} with underlying uniform local dynamics following~\cref{def:uniform_local_dynamics}. The condition number of the matrix is upper bounded by
    \begin{align}
        \kappa(A_{\pi}) \leq \frac{1+\gamma}{1-\gamma}.
    \end{align}
\end{theorem}
\begin{proof}
    We first show an upper bound to the largest singular value, using H\"older's inequality:
    \begin{align}
        \sigma_{max}(A_{\pi}) &\leq \sqrt{\| A_{\pi} \|_{\infty}} \sqrt{\| A_{\pi} \|_{1}} \\
        &\leq \sqrt{1 + \gamma} \sqrt{1 + \gamma}
    \end{align}
    The bound for the maximum row sum follows directly from the row-stochasticity of $P_{\pi}$, while the bound for the maximum column sum additionally employs the logarithmically restricted locality of the dynamics underlying $P_{\pi}$, and the maximum sum of non-diagonal column entries being $\gamma \cdot \log_2(N) \cdot \frac{1}{\log_2{N}}$.

    \noindent Next, we show a lower bound to the smallest singular value, using H\"older's inequality and~\cref{prop:row_norm_inverse,prop:col_norm_inverse}:
    \begin{align}
        \sigma_{\min}(A_{\pi}) = \frac{1}{\| A_{\pi}^{-1} \|}
        &\geq \frac{1}{\sqrt{\| A^{-1}_{\pi} \|_{\infty}} \sqrt{\| A_{\pi}^{-1} \|_{1}}} \\
        &\geq \frac{1}{\sqrt{\frac{1}{1-\gamma}} \sqrt{\frac{1}{1-\gamma}}}
    \end{align} 
    \noindent The bound for the maximum row sum follows directly from the row-stochasticity of $P_{\pi}$, while the bound for the maximum column sum exploits the uniform locality as described above. Plugging these results into the definition of the condition number reveals the stated bound.
\end{proof}

\noindent Interestingly, when assuming uniform locality the bound from~\cref{thm:uniform_local} is independent of the system size $N$ and only scales with the discount parameter $\gamma$. On the contrary, the bound for general local dynamics stated in~\cref{thm:general_local} scales as $\mathcal{O}(\sqrt{N \log N})$, assuming $\gamma$ is close to one -- which it usually is in practical setups. As $N$ scales in the qubit number as $N = 2^n$, this bound exhibits an exponential scaling. We assume it is possible to derive tighter bounds also for the general case, by incorporating statistical arguments. This might allow for applying~\cref{prop:col_norm_inverse} for bounding the smallest singular value in the proof of~\cref{thm:general_local}. However, such considerations are out of the scope of this work.

\begin{figure}[t]
    \usetikzlibrary{calc}
    \centering
    \subfigure[\label{subfig:kappa_stochastic}Environments with $\beta=0.1$ (upper plot) and $\beta=0.4$ (lower plot) exponential local dynamics following~\cref{def:exponential_local_dynamics}.]{
    \begin{tikzpicture}
        \centering
        \begin{axis}[
            name=plot2,
            ylabel=$\bm{\kappa}$,
            xtick={2, 4, 6, 8, 10, 12},
            ytick={0, 50, 100},
            yticklabels={0, 50, 100},
            xticklabels={},
            xmin=1.75,xmax=12.25,
            ymin=0.0,ymax=120.0,
            label style={font=\footnotesize},
            tick label style={font=\footnotesize},
            grid=both,
            axis x line=bottom, axis y line=left,
            width=0.99\linewidth,
            height=3cm,
            legend columns=3,
            legend style={/tikz/every even column/.append style={column sep=0.1cm, row sep=0.1cm},at={(0.9,1.65)},anchor=east,yshift=-5mm,font=\scriptsize}]
            ]

            \addplot+[draw=none,name path=A,no markers] %
            	table[x=qubits,y=std_g0.95_upper,col sep=comma]{figures/data/condition_number_perturbation=0.10.csv};
            \addplot+[draw=none,name path=B,no markers] %
            	table[x=qubits,y=std_g0.95_lower,col sep=comma]{figures/data/condition_number_perturbation=0.10.csv};
            \addplot[col_g095!40] fill between[of=A and B];
            \addplot[line width=.7pt,solid,color=col_g095,mark=o] %
            	table[x=qubits,y=avg_g0.95,col sep=comma]{figures/data/condition_number_perturbation=0.10.csv};

            \addplot+[draw=none,name path=A,no markers] %
            	table[x=qubits,y=std_g0.90_upper,col sep=comma]{figures/data/condition_number_perturbation=0.10.csv};
            \addplot+[draw=none,name path=B,no markers] %
            	table[x=qubits,y=std_g0.90_lower,col sep=comma]{figures/data/condition_number_perturbation=0.10.csv};
            \addplot[col_g090!40] fill between[of=A and B];
            \addplot[line width=.7pt,solid,color=col_g090,mark=o] %
            	table[x=qubits,y=avg_g0.90,col sep=comma]{figures/data/condition_number_perturbation=0.10.csv};

            \addplot+[draw=none,name path=A,no markers] %
            	table[x=qubits,y=std_g0.85_upper,col sep=comma]{figures/data/condition_number_perturbation=0.10.csv};
            \addplot+[draw=none,name path=B,no markers] %
            	table[x=qubits,y=std_g0.85_lower,col sep=comma]{figures/data/condition_number_perturbation=0.10.csv};
            \addplot[col_g085!40] fill between[of=A and B];
            \addplot[line width=.7pt,solid,color=col_g085,mark=o] %
            	table[x=qubits,y=avg_g0.85,col sep=comma]{figures/data/condition_number_perturbation=0.10.csv};

            \legend{,,,$\gamma=0.95$,,,,$\gamma=0.90$,,,,$\gamma=0.85$}

        \end{axis}
        \begin{axis}[
            name=plot3,
            at={($(plot2.south)+(0,-0.5cm)$)}, anchor=north,
            xlabel=$\textbf{size of state space}$,
            ylabel=$\bm{\kappa}$,
            xtick={2, 4, 6, 8, 10, 12},
             xticklabels={$2^2$, $2^4$, $2^6$, $2^8$, $2^{10}$, $2^{12}$},
           ytick={0, 20, 40},
            yticklabels={0, 20, 40},
            xmin=1.75,xmax=12.25,
            ymin=0.0,ymax=46.0,
            label style={font=\footnotesize},
            tick label style={font=\footnotesize},
            grid=both,
            axis x line=bottom, axis y line=left,
            width=0.99\linewidth,
            height=3cm,
            legend style={/tikz/every even column/.append style={column sep=0.1cm, row sep=0.1cm},at={(0.33,0.88)},anchor=east,yshift=-5mm,font=\scriptsize}]
            ]

            \addplot+[draw=none,name path=A,no markers] %
            	table[x=qubits,y=std_g0.95_upper,col sep=comma]{figures/data/condition_number_perturbation=0.40.csv};
            \addplot+[draw=none,name path=B,no markers] %
            	table[x=qubits,y=std_g0.95_lower,col sep=comma]{figures/data/condition_number_perturbation=0.40.csv};
            \addplot[col_g095!40] fill between[of=A and B];
            \addplot[line width=.7pt,solid,color=col_g095,mark=o] %
            	table[x=qubits,y=avg_g0.95,col sep=comma]{figures/data/condition_number_perturbation=0.40.csv};

            \addplot+[draw=none,name path=A,no markers] %
            	table[x=qubits,y=std_g0.90_upper,col sep=comma]{figures/data/condition_number_perturbation=0.40.csv};
            \addplot+[draw=none,name path=B,no markers] %
            	table[x=qubits,y=std_g0.90_lower,col sep=comma]{figures/data/condition_number_perturbation=0.40.csv};
            \addplot[col_g090!40] fill between[of=A and B];
            \addplot[line width=.7pt,solid,color=col_g090,mark=o] %
            	table[x=qubits,y=avg_g0.90,col sep=comma]{figures/data/condition_number_perturbation=0.40.csv};

            \addplot+[draw=none,name path=A,no markers] %
            	table[x=qubits,y=std_g0.85_upper,col sep=comma]{figures/data/condition_number_perturbation=0.40.csv};
            \addplot+[draw=none,name path=B,no markers] %
            	table[x=qubits,y=std_g0.85_lower,col sep=comma]{figures/data/condition_number_perturbation=0.40.csv};
            \addplot[col_g085!40] fill between[of=A and B];
            \addplot[line width=.7pt,solid,color=col_g085,mark=o] %
            	table[x=qubits,y=avg_g0.85,col sep=comma]{figures/data/condition_number_perturbation=0.40.csv};

        \end{axis}
        \end{tikzpicture}
        }
    \subfigure[\label{subfig:kappa_deterministic}Environment with deterministic dynamics following~\cref{def:deterministic_dynamics}.]{
    \begin{tikzpicture}
        \centering
        \begin{axis}[
            name=plot1,
            xlabel=$\textbf{size of state space}$,
            ylabel=$\textbf{condition number}~\bm{\kappa}$,
            xtick={2, 4, 6, 8, 10, 12},
            xticklabels={$2^2$, $2^4$, $2^6$, $2^8$, $2^{10}$, $2^{12}$},
            xmin=1.5,xmax=12.5,
            ymin=0.0,ymax=350.0,
            label style={font=\footnotesize},
            tick label style={font=\footnotesize},
            grid=both,
            axis x line=bottom, axis y line=left,
            width=0.99\linewidth,
            height=4cm,
            legend columns=-1,
            legend style={/tikz/every even column/.append style={column sep=0.1cm, row sep=0.1cm},at={(0.955,1.1)},anchor=east,yshift=-5mm,font=\scriptsize}]
            ]

            \addplot+[draw=none,name path=A,no markers] %
            	table[x=qubits,y=std_g0.95_upper,col sep=comma]{figures/data/condition_number_perturbation=0.00.csv};
            \addplot+[draw=none,name path=B,no markers] %
            	table[x=qubits,y=std_g0.95_lower,col sep=comma]{figures/data/condition_number_perturbation=0.00.csv};
            \addplot[col_g095!40] fill between[of=A and B];
            \addplot[line width=.7pt,solid,color=col_g095,mark=o] %
            	table[x=qubits,y=avg_g0.95,col sep=comma]{figures/data/condition_number_perturbation=0.00.csv};

            \addplot+[draw=none,name path=A,no markers] %
            	table[x=qubits,y=std_g0.90_upper,col sep=comma]{figures/data/condition_number_perturbation=0.00.csv};
            \addplot+[draw=none,name path=B,no markers] %
            	table[x=qubits,y=std_g0.90_lower,col sep=comma]{figures/data/condition_number_perturbation=0.00.csv};
            \addplot[col_g090!40] fill between[of=A and B];
            \addplot[line width=.7pt,solid,color=col_g090,mark=o] %
            	table[x=qubits,y=avg_g0.90,col sep=comma]{figures/data/condition_number_perturbation=0.00.csv};

            \addplot+[draw=none,name path=A,no markers] %
            	table[x=qubits,y=std_g0.85_upper,col sep=comma]{figures/data/condition_number_perturbation=0.00.csv};
            \addplot+[draw=none,name path=B,no markers] %
            	table[x=qubits,y=std_g0.85_lower,col sep=comma]{figures/data/condition_number_perturbation=0.00.csv};
            \addplot[col_g085!40] fill between[of=A and B];
            \addplot[line width=.7pt,solid,color=col_g085,mark=o] %
            	table[x=qubits,y=avg_g0.85,col sep=comma]{figures/data/condition_number_perturbation=0.00.csv};

            \addplot [domain=1.5:6.4,dashed,color=col_g095,line width=.7pt]{(sqrt(2^x + 0.95 * 2^x * x))*sqrt(1+0.95))/(1-0.95)};
            \addplot [domain=1.5:6.4,dashed,color=col_g090,line width=.7pt]{(sqrt(2^x + 0.95 * 2^x * x))*sqrt(1+0.9))/(1-0.9)};
            \addplot [domain=1.5:8.4,dashed,color=col_g085,line width=.7pt]{(sqrt(2^x + 0.95 * 2^x * x))*sqrt(1+0.85))/(1-0.85)};
            \addplot[color=black, dashed, domain=0.0:0.4, line width=.7pt] {0.0};


            \addplot [domain=1.5:12.4,dotted,color=col_g095,line width=.7pt]{(sqrt(1+0.95)*sqrt(1+0.95*x)*sqrt(x))/(1-0.95)};
            \addplot [domain=1.5:12.4,dotted,color=col_g090,line width=.7pt]{(sqrt(1+0.90)*sqrt(1+0.90*x)*sqrt(x))/(1-0.90)};
            \addplot [domain=1.5:12.4,dotted,color=col_g085,line width=.7pt]{(sqrt(1+0.85)*sqrt(1+0.85*x)*sqrt(x))/(1-0.85)};
                
            \addplot[color=black, dotted, domain=0.0:0.4, line width=.7pt] {0.0};


            \legend{,,,,,,,,,,,,,,,Bound from~\cref{thm:general_local},,,Empirical Bound~\cref{eq:empirical_bound}};

        \end{axis}
        \end{tikzpicture}
        }
    \subfigure[\label{subfig:kappa_uniform}Environment with uniform local dynamics following~\cref{def:uniform_local_dynamics}.]{
    \begin{tikzpicture}
        \centering
        \begin{axis}[
            name=plot1,
            xlabel=$\textbf{size of state space}$,
            ylabel=$\textbf{condition number}~\bm{\kappa}$,
            xtick={2, 4, 6, 8, 10, 12},
            xticklabels={$2^2$, $2^4$, $2^6$, $2^8$, $2^{10}$, $2^{12}$},
            xmin=1.5,xmax=12.5,
            ymin=0.0,ymax=44.0,
            label style={font=\footnotesize},
            tick label style={font=\footnotesize},
            grid=both,
            axis x line=bottom, axis y line=left,
            width=0.99\linewidth,
            height=4cm,
            legend columns=2,
            legend style={/tikz/every even column/.append style={column sep=0.1cm, row sep=0.1cm},at={(0.9,0.85)},anchor=east,yshift=-5mm,font=\scriptsize}]
            ]

            \addplot+[draw=none,name path=A,no markers] %
            	table[x=qubits,y=std_g0.95_upper,col sep=comma]{figures/data/condition_number_perturbation=1.00.csv};
            \addplot+[draw=none,name path=B,no markers] %
            	table[x=qubits,y=std_g0.95_lower,col sep=comma]{figures/data/condition_number_perturbation=1.00.csv};
            \addplot[col_g095!40] fill between[of=A and B];
            \addplot[line width=.7pt,solid,color=col_g095,mark=o] %
            	table[x=qubits,y=avg_g0.95,col sep=comma]{figures/data/condition_number_perturbation=1.00.csv};

            \addplot+[draw=none,name path=A,no markers] %
            	table[x=qubits,y=std_g0.90_upper,col sep=comma]{figures/data/condition_number_perturbation=1.00.csv};
            \addplot+[draw=none,name path=B,no markers] %
            	table[x=qubits,y=std_g0.90_lower,col sep=comma]{figures/data/condition_number_perturbation=1.00.csv};
            \addplot[col_g090!40] fill between[of=A and B];
            \addplot[line width=.7pt,solid,color=col_g090,mark=o] %
            	table[x=qubits,y=avg_g0.90,col sep=comma]{figures/data/condition_number_perturbation=1.00.csv};

            \addplot+[draw=none,name path=A,no markers] %
            	table[x=qubits,y=std_g0.85_upper,col sep=comma]{figures/data/condition_number_perturbation=1.00.csv};
            \addplot+[draw=none,name path=B,no markers] %
            	table[x=qubits,y=std_g0.85_lower,col sep=comma]{figures/data/condition_number_perturbation=1.00.csv};
            \addplot[col_g085!40] fill between[of=A and B];
            \addplot[line width=.7pt,solid,color=col_g085,mark=o] %
            	table[x=qubits,y=avg_g0.85,col sep=comma]{figures/data/condition_number_perturbation=1.00.csv};

            \addplot[color=col_g095, dashed, domain=1.5:12.4, line width=.7pt] {39.0};
            \addplot[color=col_g090, dashed, domain=1.5:12.4, line width=.7pt] {19.0};
            \addplot[color=col_g085, dashed, domain=1.5:12.4, line width=.7pt] {12.3};
            \addplot[color=black, dashed, domain=0.0:0.4, line width=.7pt] {0.0};

            \legend{,,,,,,,,,,,,,,,Bound from \cref{thm:uniform_local}};

        \end{axis}
        \end{tikzpicture}
        }
    \captionsetup{font=footnotesize}
    \caption{\label{fig:kappa}Condition number averaged over $100$ random \glspl{lse} with underlying local dynamics. Environments with exponential local dynamics are considered in (a), with a clear convergence of the condition number. The special case of deterministic dynamics is considered in (b), with the bound obtained from~\cref{thm:general_local} being extremely loose. From the data one can assume the statistical upper bound in \cref{eq:empirical_bound}. The instances with uniform local dynamics in (c) are well-conditioned for arbitrary system size as follows from~\cref{thm:uniform_local}.}
\end{figure}

We summarize the results of empirically evaluating the condition number for different environment dynamics in~\cref{fig:kappa}. For local dynamics the systems are well-conditioned, as shown in \cref{subfig:kappa_stochastic}. The loose bounds for deterministic dynamics suggest, that the estimate in \cref{eq:estimation_sigma_min} can be tightened from $\sqrt{N}$ to $\sqrt{\log_2(N)}$, resulting in the upper bound
\begin{align}
    \label{eq:empirical_bound}
    \sqrt{\log_2(N) + \gamma \log_2^2(N)} \cdot \frac{\sqrt{1 + \gamma}}{1 - \gamma}.
\end{align}%
This scaling behavior of $\mathcal{O}(\log N)$ aligns with the results in~\cref{subfig:kappa_deterministic}. We therefore suggest it is reasonable to assume, that $\kappa$ grows at most (sub-)linearly with the number of qubit -- i.e.\ logarithmically in the system size. For the local uniform dynamics in~\cref{subfig:kappa_uniform}, the bound given in~\cref{thm:uniform_local} is shown to be tight. As it is independent of the system size, there are no principled roadblocks regarding scalability. The in-between cases in~\cref{subfig:kappa_stochastic} are guaranteed the upper bound from~\cref{thm:general_local}, but with increasing value of $\beta$ seem to approach the bounds of~\cref{thm:uniform_local}. Moreover, the dependence of the loss threshold on the condition number is limited, as observed in~\cref{subsec:threshold}.

\subsection{\label{app:decomposition}Notes on Unitary Decomposition}

Executing the variational \gls{lse} solver from~\cref{subsec:lse} on quantum hardware requires access to a unitary decomposition following~\cref{eq:unitary_decomposition}. For our proof-of-concept experiments we made use of the fact that all real matrices can be decomposed into an affine combination $ A / \|A \| = X + X^t + Y - Z$, with $X, Y, Z$ unitary~\cite{Zhan_2001,Li_2002}.  It has to be noted, that explicitly determining the decomposition requires performing a \gls{svd}. Consequently, using a classical \gls{svd} technique for this construction would be an bottleneck canceling out potential quantum advantage of variational policy iteration. For hardware realizations one also needs to consider the complexity of unitary synthesis~\cite{Rietsch_2024}, which this work factored out by simulating full unitaries using the \texttt{variational-lse-solver} library~\cite{Meyer_2024}.

Closer examining and eliminating this bottleneck is out of the scope of this paper. However, it has already been shown, that quantum access is possible to the \gls{lse} underlying specific \gls{rl} environments~\cite{Cherrat_2023}. A promising tool might also be the use of Szegedy walks, which are especially suited for sparse systems~\cite{Subacsi_2019,Szegedy_2004}. Trading a larger number of decomposition terms for a trivial unitary synthesis, one might also consider tree-based approaches for Pauli decomposition~\cite{Koska_2024}.


\bibliographystyle{IEEEtran}
\bibliography{paper}


\end{document}